\documentclass[sigconf,nonacm]{acmart}
\usepackage{lipsum} %NOT ON THE LIST
\usepackage{epstopdf}
\usepackage{amsthm}
\usepackage{xcolor} %for now

\usepackage[ruled,vlined]{algorithm2e}
\usepackage{algpseudocode} %NOT ON THE LIST
\usepackage{appendix}
\newtheorem{lemma}{Lemma}

\newtheorem{problem}{Problem}

\usepackage{ulem}%NOT ON THE LIST
\usepackage{soul}
\newtheorem{theorem}{Theorem}
\usepackage{multirow} 
\usepackage[table]{xcolor}

\AtBeginDocument{\hypersetup{colorlinks=true,citecolor=purple,linkcolor=black,urlcolor=black,filecolor=black}}

\begin{document}
%\title{When Worse Routes Become Better: Preemptive Multi-Route Optimization for Age of Information}
\title{Worse Routes, Fresher Information: Preemption Redefines Which Route is Best}
 \author{Adem Utku Atasayar}
 \email{atasayar.utku@metu.edu.tr}
 \affiliation{%
   \institution{Middle East Technical University}
   \city{Ankara}
   \country{Turkiye}
 }

 \author{Aimin Li}
 \email{aimin@metu.edu.tr}
 \affiliation{%
   \institution{Middle East Technical University}
   \city{Ankara}
   \country{Turkiye}}
   
 \author{Elif Uysal}
 \email{uelif@metu.edu.tr}
 \affiliation{%
   \institution{Middle East Technical University}
   \city{Ankara}
   \country{Turkiye}
 }
\begin{abstract}
This paper studies joint sampling, route activation, and preemption control over $N$ heterogeneous transmission routes under a long-run objective combining Age of Information (AoI) with sampling and transmission costs. We formulate the problem as an average-cost impulse-control model, derive a vector-form integral average-cost optimality equation over a hybrid state space, and establish structural reductions on the state. We consider a two-route case and compute optimized policies using discretized policy iteration. Surprisingly, the results show that preemption can reverse routing preferences: \textit{the optimized policy may exclusively use a route with both a larger mean delay and a larger delay variance while avoiding the seemingly superior route.} Relative to baselines that do not allow preemption, the proposed policy reduces average AoI by up to $96\%$.
\end{abstract}
%CHECK

\maketitle

\begin{center}
\small
Accepted at the Beyond AoI Workshop, ACM MobiHoc 2026.\\
\href{https://doi.org/10.1145/3842721.3850127}
{doi:10.1145/3842721.3850127}
\end{center}

\section{Introduction}

\subsection{Background and Motivation}

Age of Information (AoI) has emerged as a fundamental metric for characterizing information freshness in real-time networked systems \cite{kosta2017age}. AoI quantifies the staleness of the most recently delivered update at the receiver and captures whether the information available at the destination remains timely enough to support {downstream} estimation
\cite{sun2020remote,tang22mobi, aiminfusion,ornee2021sampling}, feedback control \cite{soleymani2023feedback,ayan2021age}, and decision-making \cite{11393542,dong2018age}. This makes AoI particularly relevant to goal-oriented applications
\cite{9919752,10579545}, including autonomous driving \cite{kaul2011minimizing,nguyen2023vehicular}, drone and robot monitoring \cite{long2025lyapunov,liang2023age}, and real-time health monitoring \cite{ling2022age}; see \cite{yates2021age} for a comprehensive review. These goal-oriented services are increasingly deployed over integrated Terrestrial and Non-Terrestrial Networks (TN-NTN), where \textit{heterogeneous} communication routes create new opportunities for maintaining information freshness.

\textit{Routing and path diversity} have therefore become important for
controlling information freshness. In \cite{kam2016effect}, Kam \textit{et al.} showed that using multiple
paths can reduce AoI, but also introduces the challenge of
\textit{out-of-order} deliveries. Subsequent parallel-server and multi-server models quantified the benefits of
path diversity \cite{yates2018status,akar2025single}. However, most assume \textit{fixed,
prescribed} rules, such as preemptive LCFS routing over parallel paths
\cite{yates2018status} and predefined dual-server policies
\cite{akar2025single}, rather than optimizing routing decisions. 

Recent work instead treats routing as a control dimension for minimizing long-term average AoI \cite{akar2025agedependent,atasayar2025fresh,atasayar2025ageoptimal}. In \cite{atasayar2025fresh,atasayar2025ageoptimal}, joint sampling and routing are formulated over heterogeneous communication routes as a continuous-time SMDP, with a continuous waiting decision and a routing choice, leading to a threshold-based route-handover policy. In \cite{akar2025agedependent}, \textit{Akar et al.} studied discrete-time age-dependent server selection in a non-preemptive multi-server generate-at-will model with heterogeneous delay distributions and transmission costs, where waiting and server selection decisions are made only when the system is idle.

However, existing age-optimal routing formulations are developed mainly under an idealized \textit{non-preemptive}, \textit{single-active-route} setting
\cite{akar2025agedependent,atasayar2025ageoptimal,atasayar2025fresh}: 
\textit{once a route is selected for transmitting an update, the source must wait until the corresponding service is completed before making the next routing decision.} This assumption excludes two important capabilities that are relevant to modern networks: ($i$) \textit{parallel use of multiple routes} and ($ii$) proactive preemption of stale in-service updates.

Motivated by this gap, we study a continuous-time freshness-control problem in which the source jointly decides ($i$) when to sample a fresh update, ($ii$) which subset of routes to activate, allowing simultaneous use of multiple routes, and ($iii$) which in-service packets, if any, to preempt by replacing them with the fresh update. The contributions of this work are summarized below.
\subsection{Contributions of This Work}
\begin{itemize}
    \item \textbf{System model.} We introduce the first AoI-optimal routing framework that jointly permits simultaneous multi-route transmission and preemption of in-service packets, with heterogeneous delay distributions and explicit sampling and transmission costs. This generalizes prior formulations that permit at most one active route and require each transmission to complete before the next routing decision \cite{atasayar2025fresh,akar2025agedependent,atasayar2025ageoptimal} and extends \cite{li2026taming} from a single route to a multi-route setting with asynchronously evolving in-service packets.

    \item \textbf{Impulse-control formulation and solution.} Prior AoI optimization formulations commonly rely on renewal-reward processes or embedded SMDPs, with decisions structured around renewal or service-completion epochs \cite{sun2017update,atasayar2025fresh}. Allowing ongoing transmissions to be preempted at controller-chosen times requires accounting for interventions before service completion. We formulate the problem as an average-cost impulse-control problem and derive a compact vector-form integral average-cost optimality equation (ACOE) for arbitrary $N$. We establish structural reductions for packet discarding and service-age representation for arbitrary $N$, applying them in discretized policy iteration for $N=2$.

    \item \textbf{Simulations and insights.} Numerical results for $N = 2$, evaluated against a non-preemptive single-route benchmark \cite{atasayar2025fresh}, a two-route replication benchmark that we derive by extending \cite{sun2017update}, and a zero-wait baseline, demonstrate substantial freshness-cost improvements. They further show that, under preemptive operation, routing preferences can reverse, favoring a route despite its larger mean delay and delay variance.
\end{itemize}

\section{System Model and Problem Formulation}
\label{system_model}

\begin{figure}[htbp]
    \centering
    \includegraphics[width=\linewidth]{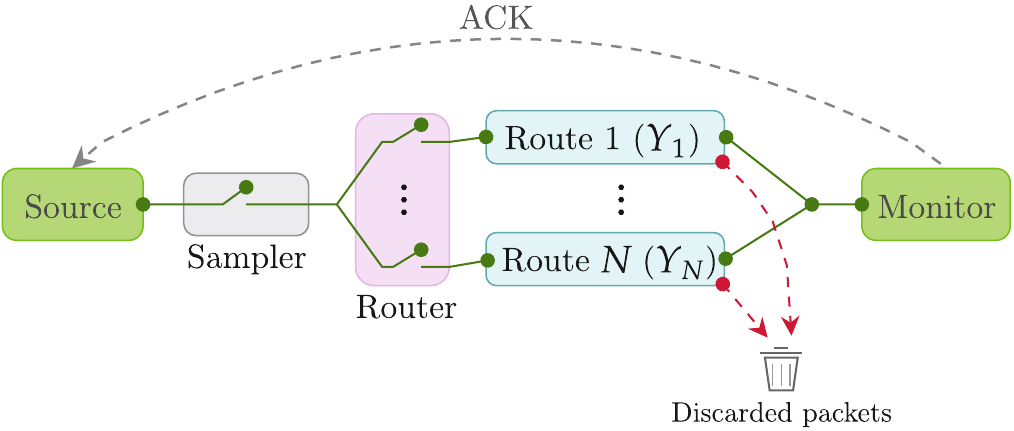}
    \caption{Status-update system with $N$ heterogeneous routes and random service times $Y_i$. A fresh update is replicated over the activated routes, replacing any packets in service.}
    \label{fig:system_model}
\end{figure}

We consider a continuous-time status-update system with a single source-destination pair and $N$ heterogeneous transmission routes, indexed by $i\in\mathcal N\triangleq\{1,\ldots,N\}$. The source aims to minimize a long-run freshness-cost objective that balances AoI against sampling and transmission costs. In contrast to existing non-preemptive, single-active-route models \cite{atasayar2025fresh,akar2025agedependent,atasayar2025ageoptimal}, our formulation allows simultaneous transmission over multiple routes and preemption of ongoing transmissions.

\subsection{Age of Information}
Let $U(t)$ denote the generation time of the freshest packet successfully delivered to the destination by time $t$. The AoI at the destination is then defined as
\begin{equation}
\Delta(t) = t - U(t), \qquad t \ge 0.
\end{equation}
\subsection{Service Model}

Each route can carry at most one packet at a time. The controller transmits only freshly sampled updates, replicating the same update over all routes selected at an intervention. Let $Y_i$ denote the random transmission delay on route~$i\in\mathcal N$. The route delays have general distributions supported on $\mathbb{R}_+$, with $\mathbb{E}[Y_i]<\infty$ and $\mathbb{E}[Y_i^2]<\infty$, and are independent across routes and successive transmissions.

\subsection{State Space}
Let $\mathbf m(t)\in\{0,1\}^N$ denote the route-occupancy vector, where $m_i(t)=1$ if route~$i$ has a packet in service and $m_i(t)=0$ otherwise. Thus, $\mathbf m(t)=\mathbf0$ represents an idle system.

For each busy route, let $b_i(t)$ denote the age of its in-service packet, which equals its elapsed service time. We set $b_i(t)=0$ for idle routes and collect these ages in $\mathbf b(t)=(b_1(t),\ldots,b_N(t))$. The hybrid state is
$
x(t)=\bigl(\Delta(t),\mathbf m(t),\mathbf b(t)\bigr).
$
We omit the time argument when referring to the current state.

\subsection{State Evolution}

Between service completions and control interventions, occupancy remains unchanged, while $\dot{\Delta}(t)=1$ and $\dot b_i(t)=m_i(t)$. When a fresh update is transmitted on route~$i$, we set $m_i(t^+)=1$ and $b_i(t^+)=0$.

When a packet on route~$i$ completes service, that route becomes idle, with $m_i(t^+)=0$ and $b_i(t^+)=0$, and the AoI becomes
\begin{equation}
\Delta(t^+)=\min\{\Delta(t^-),\, b_i(t^-)\},
\end{equation}
Here, $t^-$ and $t^+$ denote the pre- and post-jump instants. A stale completion does not affect the AoI. Other routes continue service unless their packets are discarded under the discard-after-completion rule justified in Lemma~\ref{lem:discard_on_completion}.

\subsection{Control Actions}

The controller jointly decides when to sample a fresh update, which subset of routes to activate, and which in-service packets, if any, to preempt. These decisions are made at packet-delivery epochs and controller intervention epochs.

At each decision epoch, the controller selects a control pair $(\tau,\mathbf a)$, where $\tau\geq0$ is the waiting time until the next intended intervention and $\mathbf a=(a_1,a_2,\ldots,a_N)\in \mathcal A\triangleq\{0,1\}^N\setminus\{\mathbf0\}$ is the route-activation vector. Here, $a_i=1$ means that the fresh update is transmitted on route~$i$, whereas $a_i=0$ leaves that route unchanged. If a selected route is busy, its packet in service is preempted and replaced by the fresh update.

Specifically, if the system is in state $x$ at time $t$ and the controller selects $(\tau,\mathbf a)$, the system evolves naturally on $(t,t+\tau)$. If no service completion occurs during this interval, action $\mathbf a$ is applied at time $t+\tau$. If a completion occurs earlier at some $t'\in(t,t+\tau)$, the planned intervention is canceled, the state jumps at $t'$, and a new control pair is selected.

\subsection{Problem Formulation}

The controller balances information freshness against operational costs: a sampling cost $c_s > 0$ is incurred for each fresh packet generated, and a transmission cost $c_{tx} > 0$ is incurred for each route transmission. Let $N_s(T)$ and $N_{tx}(T)$ denote the number of fresh samplings and the total number of route transmissions over $[0,T]$, respectively. An admissible policy $\pi \in \Pi$ specifies, at each decision epoch, a control pair $(\tau,\mathbf a)$ as a function of the current state. We seek a policy that minimizes the long-run average cost, as formalized below.
\begin{problem}[Long-Run Average-Cost Optimization]
\label{main_problem}
\begin{equation}
\rho^\star
\triangleq
\inf_{\pi\in\Pi}
\limsup_{T\to\infty}
\frac{1}{T}
\mathbb{E}_\pi
\left[
\int_0^T \Delta(t)\,dt
+ c_s N_s(T)
+ c_{tx} N_{tx}(T)
\right].
\end{equation}
\end{problem}
Any policy attaining the infimum in Problem~\ref{main_problem} is called optimal.

This is a long-run average-cost impulse-control problem on a hybrid state space. At each decision epoch, the controller selects both the waiting time until the next intended intervention and the corresponding transmission action, while route completions may occur earlier and trigger a new decision epoch.

\section{Average-Cost Optimality Equation}
\label{sec:acoe}

Because controller-chosen interventions compete with random route completions, the problem cannot be reduced to renewal or service-completion epochs. Accordingly, we first characterize the residual-life and first-completion kernels and then derive an integral ACOE that accounts for both event types.

\subsection{Residual-Life and First-Completion Kernels}

For the remainder of the analysis, assume that the route-delay distributions are absolutely continuous. For each $i\in\mathcal N$, let $F_i$, $f_i$, and $\bar F_i=1-F_i$ denote the CDF, density, and survival function of $Y_i$. For a service age $b_i$ satisfying $\bar F_i(b_i)>0$, the residual-life survival function and density are
\begin{equation}
\bar F_{i\mid b_i}(t)
=\frac{\bar F_i(b_i+t)}{\bar F_i(b_i)},
\qquad
f_{i\mid b_i}(t)
=\frac{f_i(b_i+t)}{\bar F_i(b_i)},
\quad t\geq 0.
\label{eq:residual_defs}
\end{equation}

For occupancy vector $\mathbf m\in\{0,1\}^N$, define the probability that no busy route completes within the next $t$ time units as
\begin{equation}
S_{\mathbf m}(t\mid\mathbf b)
\triangleq
\prod_{j=1}^{N}
\bigl[\bar F_{j\mid b_j}(t)\bigr]^{m_j}.
\label{eq:Nroute_survival}
\end{equation}
The density that Route~$i$ is the first busy route to complete at time $t$ is
\begin{equation}
q_i(t\mid\mathbf m,\mathbf b)
\triangleq
m_i f_{i\mid b_i}(t)
\prod_{j\neq i}
\bigl[\bar F_{j\mid b_j}(t)\bigr]^{m_j}.
\label{eq:Nroute_first_completion}
\end{equation}
We also define
\begin{equation}
\begin{aligned}
A_{\mathbf m}(\tau\mid\mathbf b)
&\triangleq
\int_0^\tau S_{\mathbf m}(t\mid\mathbf b)\,dt,\\
J_{\mathbf m}(\tau\mid\mathbf b)
&\triangleq
\int_0^\tau tS_{\mathbf m}(t\mid\mathbf b)\,dt.
\end{aligned}
\label{eq:Nroute_AJ}
\end{equation}

\subsection{Integral ACOE}

Let $\mathbf e_i$ denote the $i$-th standard basis vector. If Route~$i$ completes at time $t$ before the planned intervention, the post-completion state is
\begin{equation}
\Psi_i(x,t)
\triangleq
\left(
\min\{\Delta+t,b_i+t\},
\mathbf m-\mathbf e_i,
(\mathbf b+t\mathbf m)\odot(\mathbf 1-\mathbf e_i)
\right),
\label{eq:Nroute_completion_map}
\end{equation}
where $\odot$ denotes element-wise multiplication.

Suppose the controller plans to intervene after $\tau\geq0$ using activation vector $\mathbf a\in\mathcal A$. If no completion occurs before $\tau$, the state becomes $x_\tau=(\Delta+\tau,\mathbf m,\mathbf b+\tau\mathbf m)$, and the intervention transition is
\begin{equation}
\Gamma_{\mathbf a}(x_\tau)
\triangleq
\left(
\Delta+\tau,\,
\mathbf m\vee\mathbf a,\,
(\mathbf1-\mathbf a)\odot(\mathbf b+\tau\mathbf m)
\right),
\label{eq:Nroute_intervention_map}
\end{equation}
where $\vee$ denotes element-wise logical OR. The intervention cost is
\begin{equation}
K(\mathbf a)\triangleq c_s+c_{tx}\|\mathbf a\|_1.
\label{eq:Nroute_intervention_cost}
\end{equation}

Solving Problem~\ref{main_problem} is challenging because the controller must jointly choose when to intervene and how to coordinate transmissions across asynchronously evolving routes. Under general service-time distributions, completion prospects depend on the elapsed service ages, and a random completion may precede a planned intervention and alter the next decision state. These coupled dynamics require accounting for both controller-chosen interventions and random completion events. Following the average-cost impulse-control framework of \cite{stettner2022approximation}, we obtain the following integral ACOE.
\begin{theorem}[Integral ACOE]
\label{thm:Nroute_acoe}
The ACOE for the relative value function $h$ and optimal average cost $\rho^\star$ takes the form
\begin{align}
h(\Delta,\mathbf m,\mathbf b)
=
\inf_{\tau\geq0}
\Bigg\{&
(\Delta-\rho^\star)A_{\mathbf m}(\tau\mid\mathbf b)
+J_{\mathbf m}(\tau\mid\mathbf b)
\nonumber\\
&+
\sum_{i:m_i=1}
\int_0^\tau
q_i(t\mid\mathbf m,\mathbf b)
h\bigl(\Psi_i(x,t)\bigr)\,dt
\nonumber\\
&+
S_{\mathbf m}(\tau\mid\mathbf b)
\inf_{\mathbf a\in\mathcal A}
\left[
K(\mathbf a)
+h\bigl(\Gamma_{\mathbf a}(x_\tau)\bigr)
\right]
\Bigg\}.
\label{eq:Nroute_acoe}
\end{align}
\end{theorem}

\begin{proof}
See Appendix~\ref{app:proof_Nroute_acoe}.
\end{proof}

\subsection{Structural Reductions}
\label{subsec:structural_reductions}

The following properties justify the replacement-based control actions and simplify the post-completion states and numerical representation.

\begin{lemma}[Preemption and discard after completion]
\label{lem:discard_on_completion}
Assume that packet discarding incurs no cost.
\begin{enumerate}
    \item Discarding an in-service packet and transmitting a fresh replacement later offers no advantage over retaining the packet until replacement and then preempting it if it is still in service.
    \item After a service completion updates the destination AoI to $\Delta^+$, immediately discarding every remaining packet with $b_j\geq\Delta^+$ is weakly optimal.
\end{enumerate}
Here, weakly optimal means that discarding does not increase the optimal cost.
\end{lemma}

\begin{proof}
For the first claim, consider a policy that discards a packet before transmitting its replacement. Retain it while reproducing the same sampling and transmission decisions, ignoring any additional delivery when determining those decisions. Retention costs nothing, does not affect other routes or prevent replacement at the originally planned time, and any delivery of the retained packet can only improve freshness. Thus, retention is no worse.

For the second claim, a packet with $b_j\geq\Delta^+$ is no fresher than the destination information. Its age and the destination AoI increase at the same rate, while subsequent deliveries can only reduce the AoI. Hence, this packet can never improve freshness. Cost-free discarding does not prevent future transmissions and is therefore no worse.
\end{proof}

The first part justifies preemption over a separate discard-then-send action. The second part gives a discard-after-completion rule that reduces the occupancy state. For a state $x=(\Delta,\mathbf m,\mathbf b)$, define
\[
r_j(x)\triangleq m_j\mathbf 1_{\{b_j<\Delta\}},
\qquad
R(x)\triangleq
\bigl(\Delta,\mathbf r(x),\mathbf b\odot\mathbf r(x)\bigr).
\]
The reduced post-completion transition is therefore
\[
\widetilde{\Psi}_i(x,t)
\triangleq R\bigl(\Psi_i(x,t)\bigr).
\]
This transition clears the completing route and any stale remaining packets, restricting the numerical state space to $0\leq b_i\leq\Delta$ for busy routes, with equality retained to accommodate initialization. The general ACOE retains its form, with $\widetilde{\Psi}_i$ used in place of $\Psi_i$ for the reduced model.

\begin{lemma}[Full-occupancy entry and age-gap invariance]
\label{lem:full_occupancy_geometry}
Whenever the system enters full occupancy, $\mathbf m=\mathbf1$, through an admissible action, at least one service age is zero:
\[
\min_{i\in\mathcal N} b_i=0.
\]
During any interval without a completion or intervention, the service-age difference $b_j-b_i$ remains constant for every pair of busy routes $i,j$.
\end{lemma}

\begin{proof}
A completion, including subsequent stale-packet discards, cannot create full occupancy. Thus, entering or preserving full occupancy through an intervention requires a fresh transmission, resetting at least one service age to zero. Between events, each busy route satisfies $\dot b_i=1$, so
$\frac{d}{dt}(b_j-b_i)=0$.
\end{proof}

For $N=2$, the condition $\min_i b_i=0$ implies that the signed age gap $d\triangleq b_2-b_1$ determines both service ages at full-occupancy entry. Such states are therefore represented by $(\Delta,d)$, and subsequent service ages are reconstructed from $d$ and the elapsed time in the numerical method below.

\section{Numerical Solutions}
\label{sec:numerical_solutions}

The ACOE in \eqref{eq:Nroute_acoe} applies to an arbitrary number of routes. We next present a two-route case study to illustrate its numerical implementation and examine the resulting policy behavior. For this case study, the numerical procedure builds on the average-cost policy-iteration approach of \cite{li2026taming} and incorporates the structural reductions established in Section~\ref{subsec:structural_reductions}.

We truncate the age variables to $[0,y_{\max}]$, discretize them on a grid $y_g$ with spacing $\Delta t$, and retain states satisfying $b_i\leq\Delta$ for busy routes. Completion transitions use $\widetilde{\Psi}_i$ to discard stale packets. For $\mathbf m=(0,0)$, the state is indexed by $\Delta$; when $\|\mathbf m\|_1=1$, it is indexed by $(\Delta,b_i)$ for the busy route $i$.

By Lemma~\ref{lem:full_occupancy_geometry}, full-occupancy entry states are indexed by $(\Delta,d)$, where $d=b_2-b_1$. After waiting for $\theta$ without an intervening event, the service ages are
\begin{equation}
b_1(\theta;d)=(-d)^+ + \theta,
\qquad
b_2(\theta;d)=d^+ + \theta.
\label{eq:B12_reconstruct}
\end{equation}
The residual-life and first-completion quantities in \eqref{eq:residual_defs}--\eqref{eq:Nroute_AJ} are precomputed on the grids. Intervention thresholds are searched over a hybrid grid $\Theta$ that is uniform near the origin and logarithmically spaced in the tail.

For each stationary policy, the discretized ACOE forms a sparse linear policy-evaluation system, solved with a normalization condition. Policy improvement then minimizes the discretized ACOE operator over the admissible activation vectors and waiting thresholds. Following \cite{li2026taming}, we use a linear far-field closure with slope
\[
s_h \approx \mathbb E[\min\{Y_1,Y_2\}]
      = \int_0^\infty \bar F_1(t)\bar F_2(t)\,dt.
\]
The overall procedure is summarized in Algorithm~\ref{alg:two_route_pi}.

\begin{algorithm}[t]
\caption{Policy Iteration for the $N=2$ Specialization}
\label{alg:two_route_pi}
\KwIn{$y_{\max},\Delta t,\Theta,(Y_1,Y_2),(c_s,c_{tx}),\varepsilon_\rho$}
\KwOut{$(\rho,h,\pi)$}

Construct the discretized state space for $\mathbf m\in\{0,1\}^2$,
retaining $b_i\leq\Delta$ for busy routes\;
Precompute the quantities in \eqref{eq:residual_defs}--\eqref{eq:Nroute_AJ}\;
Initialize a feasible stationary policy $\pi^{(0)}$\;

\For{$k=0,1,2,\dots$}{
    \tcp{Policy evaluation}
    Form the sparse linear system corresponding to \eqref{eq:Nroute_acoe}
    under $\pi^{(k)}$, using reduced completion transitions
    $\widetilde{\Psi}_i$ and reconstruction \eqref{eq:B12_reconstruct}\;
    Impose the normalization $h(0,\mathbf 0,\mathbf 0)=0$ and the far-field linear closure\;
    Solve for $(\rho^{(k)},h^{(k)})$\;

    \tcp{Policy improvement}
    \ForEach{$x=(\Delta,\mathbf 0,\mathbf 0)$}{
        Update $\pi^{(k+1)}(x)$ by minimizing the discretized one-step operator
        over $z\in y_g$ with $z\ge \Delta$ and $\mathbf a\in\mathcal A$\;
    }
    \ForEach{$x=(\Delta,\mathbf m,\mathbf b)$ with $\mathbf m\neq\mathbf 0$}{
        Update $\pi^{(k+1)}(x)$ by minimizing the discretized one-step operator
        over $\theta\in\Theta$ and $\mathbf a\in\mathcal A$\;
    }

    \If{$\pi^{(k+1)}=\pi^{(k)}$ \textbf{or} $\bigl(k\geq1 \text{ and } |\rho^{(k)}-\rho^{(k-1)}|\leq\varepsilon_\rho\bigr)$}{
        \Return{$(\rho^{(k)},h^{(k)},\pi^{(k)})$}\;
    }
}
\end{algorithm}

\section{Simulation Results}
\label{sec:sim_results}

In this section, we evaluate the proposed policy for the $N=2$ specialization of our model against several benchmark schemes. We describe the benchmark policies and simulation settings, then present and discuss the numerical results.

\subsection{Benchmark Policies}
\label{subsec:benchmarks}

We compare the proposed policy with the following benchmark schemes.

\begin{itemize}
    \item \textbf{Single-active-route non-preemptive benchmark (JSR) \cite{atasayar2025fresh}:} This benchmark follows the JSR policy in \cite{atasayar2025fresh}. The sender jointly optimizes when to generate a new update and which route to use, but it is restricted to a single active route at any time. Once an update is transmitted on a selected route, the sender waits until its service completion before making the next decision. Thus, JSR captures optimized waiting and route selection under a non-preemptive single-active-route constraint.
    
    \item \textbf{Replicated update-or-wait benchmark (UoWV):} This benchmark is our two-route extension of the update-or-wait policy in \cite{sun2017update}. Each update is replicated over both routes and delivered when the first copy completes, after which the remaining copy is discarded. The effective service time is $\min\{Y_1,Y_2\}$, and waiting is optimized to minimize average AoI. Each update incurs one sampling cost and two transmission costs. Preemption is not allowed.

    \item \textbf{Best fixed-route zero-wait benchmark (ZW):} Under this benchmark, a new update is generated and transmitted immediately after the previous update is delivered, with no deliberate waiting and no preemption. Each update is sent on a fixed single route. For each simulation scenario, we report the better of the two fixed-route zero-wait policies, avoiding bias from an arbitrary route choice.
\end{itemize}

All policies are evaluated under the same total average-cost criterion, and their average AoI is reported to assess freshness.

\subsection{Parameter Settings}
\label{subsec:param_settings}

We consider five scenarios in which the delay of Route~1 follows a Gamma distribution and that of Route~2 follows a log-normal distribution. Let $\mu_i\triangleq\mathbb{E}[Y_i]$ and $\sigma_i\triangleq\operatorname{Std}(Y_i)$ denote the mean and standard deviation of the delay of Route~$i$. In $S_1$ and $S_5$, the route delays have different means and standard deviations, whereas $S_2$--$S_4$ match these moments to isolate the effects of distributional shape and cost. The parameters are summarized in Table~\ref{tab:sim_params}.

\begin{table}[t]
\centering
\caption{Simulation parameter settings.}
\label{tab:sim_params}
\begin{tabular}{c|cc|cc|cc}
\hline
Scenario & $\mu_1$ & $\sigma_1$ & $\mu_2$ & $\sigma_2$ & $c_s$ & $c_{\mathrm{tx}}$ \\
\hline
$S_1$ & $16$ & $\sqrt{4000}$ & $4$ & $\sqrt{1000}$ & $1$ & $1$ \\
$S_2$ & $2$ & $\sqrt{10}$ & $2$ & $\sqrt{10}$ & $1$ & $0.1$ \\
$S_3$ & $2$ & $\sqrt{10}$ & $2$ & $\sqrt{10}$ & $0.1$ & $0.1$ \\
$S_4$ & $2$ & $\sqrt{10}$ & $2$ & $\sqrt{10}$ & $0.1$ & $1$ \\
$S_5$ & $2$ & $6$ & $4$ & $\sqrt{10}$ & $5$ & $1$ \\

\hline
\end{tabular}
\end{table}
We use $\Delta t=0.1$, $y_{\max}=10$, and $|\Theta|=60$. Reported costs are estimated from independent Monte Carlo trajectories of $10^7$ events and agree with policy-iteration values within $2\%$ across all five scenarios. Halving $\Delta t$ changes the computed cost by at most $0.34\%$, while increasing $y_{\max}$ by $50\%$ changes it by less than $0.01\%$.

\subsection{Numerical Results and Discussion}
\label{subsec:numerical_results}
Fig.~\ref{fig:main_results} shows that the proposed policy achieves the lowest total average cost in all scenarios and the lowest average AoI except in $S_5$, despite optimizing the joint freshness-cost objective rather than AoI alone. Relative to JSR, it permits both simultaneous route use and preemption. In $S_1$, it uses only one route yet reduces the average cost from $32.65$ to $2.459$ and the average AoI from $32.59$ to $1.275$, highlighting the benefit of preemption-enabled route selection. The proposed policy also outperforms AoI-optimized UoWV in total cost throughout, while ZW performs worst overall.
%
%\begin{table}[h]
%\centering
%\caption{Performance comparison across policies.}
%\label{tab:main_results}
%\begin{tabular}{llccccc}
%\toprule
%\textbf{Policy} & \textbf{Metric} & $S_1$ & $S_2$ & $S_3$ & $S_4$ & $S_5$ \\
%\midrule
%          & Avg. Cost & 2.456 & 2.207 & 1.300 & 2.480 & 4.106 \\
%\rowcolor{gray!10} \multirow{-2}{*}{Proposed} & Avg. AoI  & 1.270 & 1.421 & 0.831 & 1.525 & 2.145 \\
%\midrule
%          & Avg. Cost & 32.65 & 4.966 & 4.709 & 4.966 & 8.139    \\
%\rowcolor{gray!10} \multirow{-2}{*}{JSR}      & Avg. AoI  & 32.59 & 4.652 & 4.652 & 4.652 & 6.759 \\
%\midrule
%          & Avg. Cost & 4.964 & 2.638 & 1.878 & 3.399 & 6.279 \\
%\rowcolor{gray!10} \multirow{-2}{*}{UoWV}     & Avg. AoI  & 4.184 & 1.624 & 1.624 & 1.624 & 1.996 \\
%\midrule
%          & Avg. Cost & 131.5 & 6.05  & 5.6   & 6.05  & 8.75  \\
%\rowcolor{gray!10} \multirow{-2}{*}{ZW}       & Avg. AoI  & 131   & 5.5   & 5.5   & 5.5   & 7.25  \\
%\bottomrule
%\end{tabular}
%\end{table}

\begin{figure}[t]
    \centering
    \includegraphics[width=\linewidth]{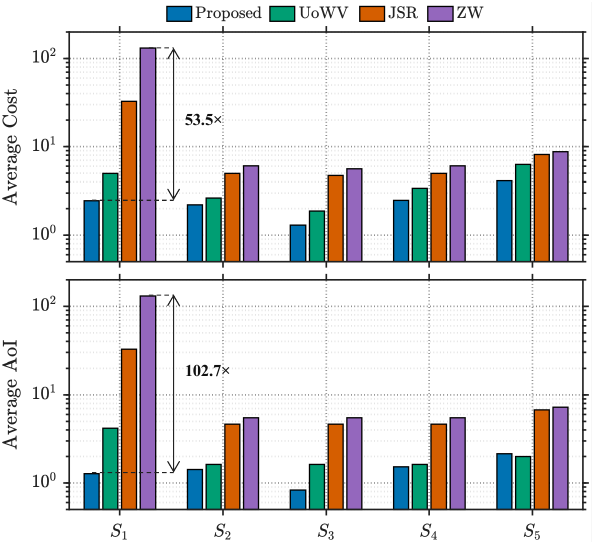}
    \caption{Performance comparison across policies. In $S_1$, ZW's average cost and AoI are $53.5\times$ and $102.7\times$ those of the proposed policy, respectively.}
    \label{fig:main_results}
\end{figure}

\begin{figure}[t]
    \centering
    \includegraphics[width=\linewidth]{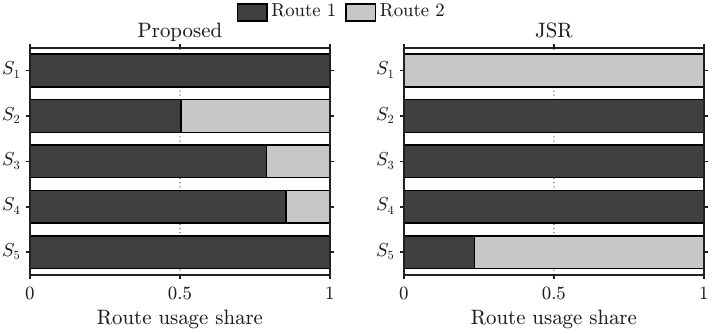}
    \caption{Route usage shares under the proposed policy and JSR, measured as fractions of transmission initiations. In $S_1$, the proposed policy uses Route~1 exclusively despite its larger mean delay and delay variance.}
    \label{fig:route_usage}
\end{figure}

Scenario~$S_1$ illustrates the paper's main qualitative message. Despite Route~1's larger mean and variance, Fig.~\ref{fig:route_usage} shows that the proposed policy uses it exclusively, whereas JSR uses only Route~2. Under preemption, the full service-time distribution matters: Route~1 places more probability mass on very small delays, allowing the proposed policy to exploit fast realizations while preempting long ones. Thus, route usefulness cannot, in general, be inferred from mean and variance alone. Relative to ZW, the proposed policy reduces average cost and AoI by factors of approximately $53.5$ and $102.7$, respectively.

Scenarios~$S_2$--$S_4$ match the routes in mean and standard deviation to isolate the effects of distributional shape and cost. As transmission becomes costlier relative to sampling, the proposed policy increasingly favors Route~1, whereas JSR uses it exclusively in all three scenarios, consistent with \cite[Lemma~7]{atasayar2025fresh}. Thus, preemption and simultaneous route use produce cost-dependent routing behavior even when the routes have identical first and second moments.

Scenario~$S_5$ represents a high-cost regime. UoWV achieves a slightly lower average AoI than the proposed policy ($1.996$ versus $2.148$), as it is more aggressive in using transmissions and does not optimize the joint freshness-cost tradeoff. However, the proposed policy reduces the total average cost from $6.279$ to $4.11$, corresponding to about a $35\%$ reduction.

\section{Conclusion}

This paper studied AoI optimization over $N$ heterogeneous routes with joint control of sampling, route activation, and preemption. We derived a vector-form integral ACOE and numerically evaluated the $N=2$ case. The results show that allowing simultaneous route usage and preemption can substantially improve freshness-cost performance. Preemption also changes routing preferences, enabling a route that appears inferior in both mean delay and delay variance to become the preferred choice.

\appendix
\section{Proof of Theorem~\ref{thm:Nroute_acoe}}
\label{app:proof_Nroute_acoe}

\begin{proof}
Fix a state $x=(\Delta,\mathbf m,\mathbf b)$ and a planned intervention time $\tau$. By independence, $S_{\mathbf m}(t\mid\mathbf b)$ is the probability that no busy route completes by time $t$, while $q_i(t\mid\mathbf m,\mathbf b)$ is the density that Route~$i$ completes first at time $t$. Simultaneous completions have probability zero because the service-time distributions are absolutely continuous.

Let $T(x)$ denote the first completion time among the currently busy routes, with $T(x)=+\infty$ when $\mathbf m=\mathbf0$. Until the first completion or the planned intervention, the AoI is $\Delta+t$. Hence, the expected running cost relative to $\rho^\star$ is
\begin{equation}
\begin{aligned}
\mathbb E\!\left[
\int_0^{\min\{\tau,T(x)\}}
(\Delta(t)-\rho^\star)\,dt
\right]
&=
\int_0^\tau
(\Delta+t-\rho^\star)
S_{\mathbf m}(t\mid\mathbf b)\,dt\\
&=
(\Delta-\rho^\star)A_{\mathbf m}(\tau\mid\mathbf b)
+J_{\mathbf m}(\tau\mid\mathbf b).
\end{aligned}
\label{eq:Nroute_running_cost}
\end{equation}
If a busy route $i$ with $m_i=1$ completes first at time $t<\tau$, the state becomes $\Psi_i(x,t)$, yielding the continuation term
\[
\int_0^\tau q_i(t\mid\mathbf m,\mathbf b)
h\bigl(\Psi_i(x,t)\bigr)\,dt.
\]
Summing over the busy routes $i:m_i=1$ accounts for every possible completion before $\tau$; the sum is empty when $\mathbf m=\mathbf0$.

With probability $S_{\mathbf m}(\tau\mid\mathbf b)$, no route completes before $\tau$. The controller then selects $\mathbf a\in\mathcal A$, incurs $K(\mathbf a)$, and moves to $\Gamma_{\mathbf a}(x_\tau)$. The optimal conditional continuation cost on this event is therefore
\[
\inf_{\mathbf a\in\mathcal A}
\left[K(\mathbf a)+h\bigl(\Gamma_{\mathbf a}(x_\tau)\bigr)\right].
\]
Combining the running, completion, and intervention terms and minimizing over $\tau\geq0$ yields \eqref{eq:Nroute_acoe}.
\end{proof}

\begin{acks}
    This work is supported by the European Union through ERC Advanced Grant 101122990-GO SPACE-ERC-2023-A and The Scientific and Technological Research Council of Türkiye (TÜBİTAK) 1515 Frontier R\&D Laboratories Support Program for Türk Telekom 6G R\&D Lab under project number 5249902. Views and opinions expressed are however those of the author(s) only and do not necessarily reflect those of the European Union or the European Research Council Executive Agency. Neither the European Union nor the granting authority can be held responsible for them.
\end{acks}

\bibliographystyle{ACM-Reference-Format}
\bibliography{references}

\appendix
\end{document}